\documentclass[11pt,a4paper]{article}

\usepackage{amsmath,amsfonts,amssymb,amsthm}
\usepackage{amsthm}
\usepackage{graphicx,color}
\usepackage{boxedminipage}
\usepackage[ruled,boxed,linesnumbered]{algorithm2e}
\usepackage{framed}
\usepackage{enumitem}
\usepackage{thmtools}
\usepackage{thm-restate}
\usepackage{xspace}
\usepackage{bm}
\usepackage{todonotes}
\usepackage{footnote}
\usepackage{mathtools}
\usepackage{mathrsfs}
\usepackage{vmargin}
\setmarginsrb{2.3cm}{2.3cm}{2.3cm}{2.3cm}{0pt}{0pt}{0pt}{6mm}
\usepackage{todonotes}
\usepackage[most]{tcolorbox}
\usepackage{footnote}
\usepackage{rotating}
 \usepackage[pdftex, plainpages = false, pdfpagelabels, 
                 bookmarks=true,
                 bookmarksopen = true,
                 bookmarksnumbered = true,
                 breaklinks = true,
                 linktocpage,
                 pagebackref,
                 colorlinks = true,  % was true
                 linkcolor = blue,
                 urlcolor  = blue,
                 citecolor = red,
                 anchorcolor = green,
                 hyperindex = true,
                 hyperfigures
                 ]{hyperref} 
 \usepackage{xifthen}
 \usepackage{tabularx}
\usepackage{tikz} 
 \usetikzlibrary{calc}
 \usepackage{thm-autoref}
\usepackage[nameinlink]{cleveref}

\usepackage{soul}

\newtheorem{theorem}{Theorem}
\newtheorem{lemma}{Lemma}
\newtheorem{claim}{Claim}
\newtheorem{corollary}{Corollary}

\newtheorem{observation}{Observation}
\newtheorem{proposition}{Proposition}

\theoremstyle{definition}

\newlength{\RoundedBoxWidth}
\newsavebox{\GrayRoundedBox}
\newenvironment{GrayBox}[1]%
   {\setlength{\RoundedBoxWidth}{.93\textwidth}
    \def\boxheading{#1}
    \begin{lrbox}{\GrayRoundedBox}
       \begin{minipage}{\RoundedBoxWidth}}%
   {   \end{minipage}
    \end{lrbox}
    \begin{center}
    \begin{tikzpicture}%
       \node(Text)[draw=black!20,fill=white,rounded corners,%
             inner sep=2ex,text width=\RoundedBoxWidth]%
             {\usebox{\GrayRoundedBox}};
        \coordinate(x) at (current bounding box.north west);
        \node [draw=white,rectangle,inner sep=3pt,anchor=north west,fill=white] 
        at ($(x)+(6pt,.75em)$) {\boxheading};
    \end{tikzpicture}
    \end{center}}

\newenvironment{defproblemx}[2][]{\noindent\ignorespaces%
                                \FrameSep=6pt%
                                \parindent=0pt%
                \vspace*{-1.5em}
                \ifthenelse{\isempty{#1}}{%
                  \begin{GrayBox}{\textsc{#2}}%                
                }{%
                  \begin{GrayBox}{\textsc{#2} parameterized by~{#1}}%  
                }
                \begin{tabular*}{\textwidth}{@{\hspace{.1em}} >{\itshape} p{1.8cm} p{0.8\textwidth} @{}}%        
            }{
                \end{tabular*}%
                \end{GrayBox}%
                \ignorespacesafterend
            }  

\newenvironment{defproblemxb}[2][]{\noindent\ignorespaces%problemx with less spacing 
  \FrameSep=6pt%
  \parindent=0pt%
  \vspace*{-1.5em}
  \ifthenelse{\isempty{#1}}{%
    \begin{GrayBox}{\textsc{#2}}%                
    }{%
      \begin{GrayBox}{\textsc{#2} parameterized by~{#1}}%  
      }
      \begin{tabular*}{\textwidth}{@{\hspace{.1em}} >{\itshape} p{1.2cm} p{0.85\textwidth} @{}}%        
      }{
      \end{tabular*}%
    \end{GrayBox}%
    \ignorespacesafterend
  }

\newcommand{\lr}[1]{\left( #1\right)}
\newcommand{\LR}[1]{\left\{ #1\right\}}

\renewcommand{\d}{\mathsf{d}}
\newcommand{\cost}{\mathsf{cost}}
\newcommand{\wcost}{\mathsf{wcost}}
\newcommand{\cI}{\mathcal{I}}
\newcommand{\diam}{\mathsf{diam}}
\newcommand{\summ}[1]{\mathsf{sum}_{\sim #1}}
\newcommand{\deltam}[1]{\cost_{#1}}
\newcommand{\deltamw}[1]{\wcost_{#1}}
\newcommand{\real}{\mathbb{R}}
\newcommand{\opt}{\mathsf{OPT}}
\newcommand{\xij}{X_{i, j}}
\newcommand{\sij}{S_{i, j}}

\newcommand{\kmed}{\textsc{$k$-Median}\xspace}
\newcommand{\kmeans}{\textsc{$k$-Means}\xspace}
\newcommand{\matmed}{\textsc{Matroid Median}\xspace}
\newcommand{\colorkmed}{\textsc{Colorful $k$-Median}\xspace}
\newcommand{\matmedwo}{\textsc{Matroid Median with Outliers}\xspace}
\newcommand{\kmedwo}{\textsc{$k$-MedianOut}\xspace}
\newcommand{\kmeanswo}{\textsc{$k$-MeansOut}\xspace}
\newcommand{\kzwo}{\textsc{$(k, z)$-Clustering with Outliers}\xspace}

\begin{document}

\title{Clustering What Matters:\\ Optimal Approximation for Clustering with Outliers}

\author{Akanksha Agrawal\thanks{
Indian Institute of Technology Madras, Chennai, India.}
\and
Tanmay Inamdar\thanks{University of Bergen, Norway}
\and
Saket Saurabh\addtocounter{footnote}{-1}\footnotemark{} \thanks{The Institute of Mathematical Science, HBNI, Chennai, India, and University of Bergen, Norway}
\and
Jie Xue~\thanks{New York University Shanghai, China} 
}

%\author{Anonymous Submission}
\date{}

\maketitle

\begin{abstract}
Clustering with outliers is one of the most fundamental problems in Computer Science.  Given a set $X$ of $n$ points and two numbers $k,m$, the clustering with outliers aims to exclude $m$ points from $X$ and partition the remaining points into $k$ clusters that minimizes a certain cost function. In this paper, we give a general approach for solving clustering with outliers, which results in a fixed-parameter tractable (FPT) algorithm in $k$ and $m$ \footnote{That is, an algorithm with running time of the form $f(k, m) \cdot n^{O(1)}$ for some function $f$.}, that almost matches the approximation ratio for its outlier-free counterpart.
As a corollary, we obtain FPT approximation algorithms with optimal approximation ratios for $k$-\textsc{Median} and $k$-\textsc{Means} with outliers in general and Euclidean metrics. We also exhibit more applications of our approach to other variants of the problem that impose additional constraints on the clustering, such as fairness or matroid constraints. 
\end{abstract}

%!TEX root = mainpaper.tex
\section{Introduction}
Clustering is a family of problems that aims to group a given set of objects in a meaningful way---the exact ``meaning'' may vary based on the application. These are fundamental problems in Computer Science with applications ranging across multiple fields like pattern recognition, machine learning, computational biology, bioinformatics and social science. Thus, these problems have been a subject of extensive studies in the field of Algorithm Design (and its sub-fields), see for instance, the surveys on this topic (and references therein)~\cite{xu2015comprehensive,rokach2009survey,blomer2016theoretical}. 

Two of the central clustering problems are \kmed\ and \kmeans. In the standard \kmed problem, we are given a set $X$ of $n$ points, and an integer $k$, and the goal is to find a set $C^* \subseteq X$ of at most $k$ \emph{centers}, such that the following cost function is minimized over all subsets $C$ of size at most $k$.
\[ \cost(X, C) \coloneqq \sum_{p \in X} \min_{c \in C} \d(p, c) \]
In \kmeans, the objective function instead contains the sum of squares of distances.

Often real world data are contaminated with a small amount of noise and these  noises can substantially change the clusters that we obtain using the underlying algorithm. To circumvent the issue created by such noises, there are several studies of clustering problems with outliers, see for instance,~\cite{Chen08,KrishnaswamyLS18,GoyalJ020,FengZHXW19,friggstad2019approximation,AlmanzaEPR22}. 

In outlier extension of the $k$-\textsc{Median} problem, which we call \kmedwo, we are also given an additional integer $m \ge 0$ that denotes the number of \emph{outliers} that we are allowed to drop. We want to find a set $C$ of at most $k$ centers, and a set $Y \subseteq X$ of at most $m$ outliers, such that $\cost(X \setminus Y, C) \coloneqq \sum_{p \in X \setminus Y} \min_{c \in C} \d(p, c)$ is minimized over all $(Y, C)$ satisfying the requirements. Observe that the cost of clustering for \kmedwo equals the sum of distances of each point to its nearest center, after excluding a set of $m$ points from consideration \footnote{Our results actually hold for a more general formulation of \kmed, where the set of candidate centers may be different from the set $X$ of points to be clustered. We consider this general setting in the technical sections.}. We remark that in a similar spirit we can define the outlier version of the \kmeans problem, which we call \kmeanswo.

In this paper, we will focus on approximation algorithms. An algorithm is said to have an approximation ratio of $\alpha$, if it is guaranteed to return a solution of cost no greater than $\alpha$ times the optimal cost, while satisfying all other conditions. That is, the solution must contain at most $k$ centers, and drop $m$ outliers. If the algorithm is randomized, then it must return such a solution with \emph{high probability}, i.e., probability at least $1-n^{-c}$ for some $c \ge 1$.

For a fixed set $C$ of centers, the set of $m$ outliers is automatically defined, namely the set of $m$ points that are farthest from $C$ (breaking ties arbitrarily). Thus, an optimal clustering for \kmedwo, just like \kmed, can be found in $n^{O(k)}$ time by enumerating all center sets. On the other hand, we can enumerate all $n^{O(m)}$ subsets of outliers, and reduce the problem directly to \kmed. Other than these straightforward observations, there are several non-trivial approximations known for \kmedwo, which we discuss in a subsequent paragraph.

\paragraph{Our Results.} In this work, we describe a general framework that reduces a \emph{clustering with outliers} problem (such as \kmedwo or \kmeanswo) to its outlier-free counterpart in an approximation-preserving fashion. More specifically, given an instance $\cI$ of \kmedwo, our reduction runs in time $f(k, m, \epsilon) \cdot n^{O(1)}$, and produces multiple instances of \kmed, such that a $\beta$-approximation for \emph{at least} one of the produced instances of \kmed implies a $(\beta+\epsilon)$-approximation for the original instance $\cI$ of \kmedwo. This is the main result of our paper.

Our framework does not rely on the specific properties of the underlying metric space. Thus, for special metrics, such as Euclidean spaces, or shortest-path metrics induced by sparse graph classes, for which FPT $(1+\epsilon)$-approximations are known for \kmed, our framework implies matching approximation for \kmedwo. Finally, our framework is quite versatile in that one can extend it to obtain approximation-preserving FPT reductions for related \emph{clustering with outliers} problems, such as \kmeanswo, and clustering problems with \emph{fair outliers} (such as \cite{Bandyapadhyay0P19,JiaSS20}), and \matmedwo. We conclude by giving a partial list of the corollaries of our reduction framework. The running time of each algorithm is $f(k, m, \epsilon) \cdot n^{O(1)}$ for some function $f$ that depends on the problem and the setting. Next to each result, we also cite the result that we use as a black box to solve the outlier-free clustering problem.
\begin{itemize}
	\item $(1+\frac{2}{e}+\epsilon) \approx (1.74+\epsilon)$-approximation (resp. $1+\frac{8}{e}+\epsilon)$-approximation) for \kmedwo (resp.\ \kmeanswo) in general metrics \cite{Cohen-AddadSS21}. These approximations are tight even for $m = 0$, under a reasonable complexity theoretic hypothesis, as shown in the same paper.
	\item $(1+\epsilon)$-approximation for \kmedwo and \kmeanswo in (i) metric spaces of constant doubling dimensions, which includes Euclidean spaces of constant dimension, (ii) metrics induced by graphs of bounded treewidth, and (iii) metrics induced by graphs that exclude a fixed graph as a minor (such as planar graphs). \cite{Cohen-AddadSS21}.
	\item $(2+\epsilon)$-approximation for \matmedwo in general metrics, where $k$ refers to the \emph{rank} of the matroid. \cite{DBLP:conf/icalp/Cohen-AddadG0LL19}
	\item $(1+\frac{2}{e}+\epsilon)$-approximation for \colorkmed in general metrics, where $m$ denotes the total number of outliers across all color classes \cite{DBLP:conf/icalp/Cohen-AddadG0LL19}. The preceding two problems are orthogonal generalizations of \kmedwo, and are formally defined in Section \ref{sec:extensions}.
\end{itemize}

\paragraph{Our Techniques.} Our reduction is inspired from the following seemingly simple observation that relates \kmedwo and \kmed. Let $\cI$ be an instance of \kmedwo, where we want to find a set $C$ of $k$ centers, such that the sum of distances of all except at most $m$ points to the nearest center in $C$ is minimized. By treating the outliers in an optimal solution for $\cI$ as \emph{virtual centers}, one obtains a solution for $(k+m)$-\textsc{Median} \emph{without outliers} whose cost is at most the optimal cost of $\cI$. In other words, the optimal cost of an appropriately defined instance $\widetilde{\cI}$ of $(k+m)$-\textsc{Median} is a \emph{lower bound} on the optimal cost of $\cI$. Since \kmed is a well-studied problem, at this point, one would hope that it is sufficient to restrict the attention to $\widetilde{\cI}$. That is, if we obtain a solution (i.e., a set of $k+m$ centers) for $\widetilde{\cI}$, can then be modified to obtain a solution (i.e., a set of $k$ centers and $m$ outliers) for $\cI$. However, it is unclear whether one can do such a modification without blowing up the cost for $\cI$. Nevertheless, this connection between $\widetilde{\cI}$ and $\cI$ turns out to be useful, but we need several new ideas to exploit it. 

As in before, we start with a constant approximation for $\widetilde{\cI}$, and perform a sampling similar to \cite{chen2009coresets} to obtain a weighted set of points. This set is obtained by dividing the set of points connected to each center in the approximate solution into \emph{concentric rings}, such that the ``error'' introduced in the cost by treating all points in the ring as identical is negligible. Then, we sample $O((k+m) \log n/\epsilon)$ points from each ring, and give each point an appropriate weight. We then prove a crucial concentration bound (cf. Lemma \ref{lem:kmed-samplinglemma}), which informally speaking relates the connection cost of original set of points in a ring, and the corresponding weighted sample. In particular, for \emph{any set of $k$ centers}, with good probability, the difference between the original and the weighted costs is ``small'', \emph{even after excluding at most $m$ outliers from both sets}. Intuitively speaking, this concentration bound holds because the sample size is large enough compared to both $k$ and $m$. Then, by taking the union of all such samples, we obtain a weighted set $S$ of $O(((k+m)\log n/\epsilon)^2)$ points that preserves the connection cost to any set of $k$ centers, even after excluding $m$ outliers with at least a constant probability. Then, we enumerate all sets $Y$ of size $m$ from $S$, and solve the resulting \kmed instance induced on $S \setminus Y$. Finally, we argue that at least one of the resulting instances $\cI'$ will have the property that, a $\beta$-approximation for $\cI'$ implies a $(\beta+\epsilon)$-approximation for $\cI$. 

\paragraph{Related Work.} The first constant approximation for \kmedwo was given by \cite{Chen08} for some large constant. More recently, \cite{KrishnaswamyLS18, GuptaMZ21} gave constant approximations based on iterative LP rounding technique, and the $6.387$-approximation by latter is currently the best known approximation. These approximation algorithms run in polynomial time in $n$. \cite{KrishnaswamyLS18} also give the best known polynomial approximations for related problems of \kmeanswo and \textsc{Matroid Median}. 

Now we turn to FPT approximations, which is also the setting for our results. To the best of our knowledge, there are three works in this setting, \cite{FengZHXW19,GoyalJ020,StatmanRF20}. The idea of relating  $k$-\textsc{Median with $m$ Outliers} to $(k+m)$-\textsc{Median} that we discuss above is also present in these works. Even though it is not stated explicitly, the approach of \cite{StatmanRF20} can be used to obtain FPT approximations in general metrics; albeit with a worse approximation ratio. However, by using additional properties of Euclidean \kmedwo/\kmeanswo (where one is allowed to place centers anywhere in $\mathbb{R}^d$) their approach yields a $(1+\epsilon)$-approximation in FPT time. \cite{GoyalJ020} design approximation algorithms with ratio of $3+\epsilon$ for \kmedwo (resp.\ $9+\epsilon$ for \kmedwo) in time $((k+m)/\epsilon)^{O(k)} \cdot n^{O(1)}$. Thus, our approximation ratios of $1+\frac{2}{e}+\epsilon$ for \kmeanswo, and $1+\frac{8}{e}+\epsilon$ for \kmeanswo improve on these results -- albeit with a slightly worse FPT running time. Furthermore, our result is essentially an \emph{approximation-preserving reduction} from \kmedwo to \kmed in the same metric, which yields $(1+\epsilon)$-approximations in some special settings as discussed earlier. On the other hand, it seems that a loss of $3+\epsilon$ (resp.\ $9+\epsilon$) in the approximation guarantee is inherent to the algorithm of \cite{GoyalJ020}.

On the lower bound side, \cite{GuhaK99} showed it is NP-hard to approximate \kmed (and thus \kmedwo) within a factor $1+\frac{2}{e}-\epsilon$ for any $\epsilon > 0$. Recently, \cite{DBLP:conf/icalp/Cohen-AddadG0LL19} strengthened this result assuming Gap-ETH, and showed that an $(1+\frac{2}{e}-\epsilon)$-approximation algorithm must take at least $n^{k^{g(\epsilon)}}$ time for some function $g()$. 

\textit{Bicriteria approximations} relax the strict requirement of using at most $k$ centers, or dropping at most $m$ outliers, in order to give improved approximation ratios, or efficiency (or both). For \kmedwo, \cite{CharikarKMN01} gave a $4(1+1/\epsilon)$-approximation, while dropping $m(1+\epsilon)$ outliers. \cite{GuptaKLMV17} gave a constant approximation based on local search for \kmeanswo that drops $O(km \log(n\Delta))$ outliers, where $\Delta$ is the diameter of the set of points. \cite{FriggstadRS19} gave a $(25+\epsilon)$-approximation that uses $k(1+\epsilon)$ centers but only drops $m$ outliers. In Euclidean spaces, they also give a $(1+\epsilon)$-approximation that returns a solution with $k(1+\epsilon)$ centers. 
%!TEX root = mainpaper.tex

\section{Preliminaries}

\paragraph{Basic notions.}
Let $(\Gamma, \d)$ be a metric space, where $\Gamma$ is a finite set of points, and $\d: \Gamma \times \Gamma \to \real$ is a distance function satisfying symmetry and triangle inequality. For any finite set $S \subseteq \Gamma$ and a point $p \in \Gamma$, we let $\d(p, S) \coloneqq \min_{s \in S} \d(p, S)$, and let $\diam(S) \coloneqq \max_{x, y \in S} \d(x, y)$. For two non-empty sets $S, C \subseteq \Gamma$, let $\d(S, C) = \min_{p \in S} \d(p, S) = \min_{p \in S} \min_{c \in C} \d(p, c)$. For a point $p \in \Gamma, r \ge 0$, and a set $C \subseteq \Gamma$, let $B_C(p, r) = \{q \in C: \d(p, c) \le r\}$. Let $T$ be a finite (multi)set of $n$ real numbers, for some positive integer $n$, and let $1 \le m \le n$. Then, we use the notation $\summ{m}(T)$ to denote the sum of $n-m$ smallest values in $T$ (including repetitions in case of a multi-set). 

\paragraph{The $k$-median problem.}
In the \kmed problem, an instance is a triple $\cI = (X, F, k)$, where $X$ and $F$ are finite sets of points in some metric space $(\Gamma,\d)$, and $k \geq 1$ is an integer. The points in $X$ are called \textit{clients}, and the points in $F$ are called \emph{facilities} or \emph{centers}.
%We use $\d: (X \cup F) \times (X \cup F) \rightarrow \mathbb{R}_{\geq 0}$ to denote the metric on $X \cup F$ (induced by the underlying metric space of $X$ and $F$).
The task is to find a subset $C \subseteq F$ of size at most $k$ that minimizes the cost function
\[\cost(X, C) \coloneqq \sum_{p \in X} \d(p, C). \]
%In this paper, the underlying metric space of $X$ and $F$ is either a general metric space (in which case $F$ must be finite and the metric $\d$ is explicitly given in the input) or the Euclidean space $\mathbb{R}^d$ for a fixed $d$ (in which case $F$ can be either finite or the entire space $\mathbb{R}^d$, and the input specifies the coordinates of the points which induce the metric $\d$).
The \textit{size} of an instance $\cI = (X, F, k)$ is defined as $|\cI| = |X \cup F|$, which we denote by $n$. 
%if $F$ is a finite set, and $|\cI| = |X|$ if $F = \mathbb{R}^d$ in the Euclidean case.

\paragraph{$k$-median with outliers.}
The input to \kmedwo contains an additional integer $0 \le m \le n$, and thus an instance is given by a $4$-tuple $\mathcal{I} = (X, F, k, m)$.
Let $C \subseteq F$ be a set of facilities. We define $\deltam{m}(X, C) \coloneqq \summ{m} \{\cost(p, C) : p \in X \}$, i.e., the sum of $n-m$ smallest distances of points in $X$ to the set of centers $C$. The goal is to find a set of centers $C$ minimizing $\deltam{m}(X, C)$ over all sets $C \subseteq F$ of size at most $k$. Given a set $C \subseteq F$ of centers, we denote the corresponding solution by $(Y, C)$, where $Y \subseteq X$ is a set of $m$ outlier points in $X$ with largest distances realizing $\deltam{m}(X, C)$. Given an instance $\cI$ of \kmedwo, we use $\opt(\cI)$ to denote the value of an optimal solution to $\cI$.

\paragraph{Weighted sets and random samples.}
During the course of the algorithm, we will often deal with \emph{weighted sets} of points. Here, $S \subseteq X$ is a weighted set, with each point $p \in S$ having integer weight $w(p) \ge 0$. For any set $C \subseteq F$ and $1 \le m \le |S|$, define $\deltamw{m}(S, C) \coloneqq \summ{m}\{  d(p, C) \cdot w(p) : p \in S\}$.
A \textit{random sample} of a finite set $S$ refers to a random subset of $S$.
Throughout this paper, random samples are always generated by picking points \textit{uniformly} and \textit{independently}.

%The notion of $k$-median with outliers naturally extends to weighted sets, where we seek to find a set $C \subseteq F$ minimizing $\deltamw{m}(S, C)$. 
%!TEX root = mainpaper.tex

\section{$k$-Median with Outliers} \label{sec:kmed-general}

In this section, we give our FPT reduction from \kmedwo to the standard \kmed problem. Formally, we shall prove the following theorem.

\begin{theorem} \label{thm-generalreduction}
	Suppose there exists a $\beta$-approximation algorithm for \kmed with running time $T(n,k)$, and a $\tau$-approximation algorithm for $k+m$-\textsc{Median} with polynomial running time, where $\beta$ and $\tau$ are constants. Then there exists a $(\beta+\epsilon)$-approximation algorithm for \kmedwo with running time $\lr{\frac{k+m}{\epsilon}}^{O(m)} \cdot T(n,k) \cdot n^{O(1)}$, where $n$ is the instance size and $m$ is the number of outliers.
\end{theorem}

\noindent
Combining the above theorem with the known $(1+\frac{2}{e}+\epsilon)$-approximation $k$-median algorithm \cite{DBLP:conf/icalp/Cohen-AddadG0LL19} that runs in $(k/\epsilon)^{O(k)} \cdot n^{O(1)}$ time, we directly have the following result.

\begin{corollary}
	There exists a $(1+\frac{2}{e}+\epsilon)$-approximation algorithm for \kmedwo with running time $\lr{\frac{k+m}{\epsilon}}^{O(m)} \cdot \lr{\frac{k}{\epsilon}}^{O(k)} \cdot n^{O(1)}$, where $n$ is the instance size and $m$ is the number of outliers.
\end{corollary}

The rest of this section is dedicated to proving Theorem~\ref{thm-generalreduction}.
%Suppose there exists a $\beta$-approximation algorithm for \kmed with running time $T(n,k)$.
Let $\cI = (X, F, k, m)$ be an instance of \kmedwo.
We define a $(k+m)$-\textsc{Median} instance $\cI' = (X, F \cup X, k+m)$, where in addition to the original set of facilities, there is a facility co-located with each client.
We have the following observation.

\begin{observation}\label{obs:kmed-lb}
	$\opt(\cI') \le \opt(\cI)$, i.e., the value of an optimal solution to $\cI'$ is a lower bound on the value of an optimal solution to $\cI$.
\end{observation}
\begin{proof}
	Let $(Y^*, C^*)$ be an optimal solution to $\cI$ realizing the value $\opt(\cI)$. We define a solution $(Y', C')$ for $\cI'$ as follows: let $Y' = X$, and $C' = C^* \cup  Y^*$. That is, the set of centers $C'$ is obtained by adding a facility co-located with each outlier point from $Y^*$. Now we argue about the costs. Since $C^* \subseteq C'$, for each point $p \in Y^*$, $\d(p, C') \le \d(p, C^*)$. On the other hand, for each $q \in X \setminus Y^*$, $\d(q, C') = 0$, since there is a co-located center in $C^*$. This implies that $\deltam{0}(X, C') \le \deltam{m}(X, C)$. Since the solution $(Y', C')$ is feasible for the instance $\cI'$, it follows that $\opt(\cI')$ is no larger than the cost $\deltam{0}(X, C')$.
\end{proof}

Now, we use $\tau$-approximation algorithm guaranteed by the theorem, for the instance $\cI'$, and obtain a set of at most $k' \le k+m$ centers $A$ such that $\deltam{0}(X, A) \le \tau \cdot \opt(\cI') \le \tau \cdot \opt(\cI)$.
By assumption, running this algorithm takes polynomial time.
Let $R = \frac{\deltam{0}(X, A)}{\tau n}$ be a lower bound on average radius, and $\phi = \lceil \log(\tau n) \rceil$. 
For each center $c_i \in A$, let $X_i \subseteq X$ denote the set of points whose closest center in $A$ is $c_i$.
By arbitrarily breaking ties, we can assume that the sets $X_i$ are disjoint, i.e., $\{X_i\}_{1 \le i \le k'}$ forms a partition of $X$. 
%Now, we define the set of \emph{rings} centered at each center $c_i$ as follows.
Now we further partition each $X_i$ into smaller groups such that the points in each group have similar distances to $c_i$.
Specifically, we define
\[\xij \coloneqq \begin{cases}
	B_{X_i}(c_i, R) & \text{ if } j = 0,
	\\B_{X_i}(c_i, 2^j R) \setminus B_{X_i}(c_i, 2^{j-1} R) & \text{ if } j \ge 1.
\end{cases}
\]

Let $s = \frac{c\tau^2}{\epsilon^2} \lr{m+ k \ln n  + \ln(1/\lambda)}$, for some large enough constant $c$.
We define a weighted set of points $\sij \subseteq \xij$ as follows.
If $|\xij| \le s$, then we say $\xij$ is \emph{small}.
In this case, define $\sij \coloneqq \xij$ and let the weight $w_{i, j}$ of each point $p \in \sij$ be $1$.
Otherwise, $|\xij| > s$ and we say that $\xij$ is \emph{large}.
In this case, we take a random sample $S_{i, j} \subseteq X_{i, j}$ of size $s$.
%we sample $s$ points at random (with replacement), and let $S_{i, j} \subseteq X_{i, j}$ be the resulting sample of size $s$.
We set the weight of every point in $S_{i,j}$ to be $w_{i, j} = |X_{i, j}|/|S_{i, j}|$. 
For convenience, assume the weights $w_{i, j}$ to be integers \footnote{We defer the discussion on how to ensure the integrality of the weights to Section \ref{subsec:intweights}.}.
\iffalse
In this case, let $Y_{i, j} \subseteq X_{i, j}$ be an arbitrary subset of size $|\xij| \mod s$.
We add each point $q \in Y_{i, j}$ to $\sij$ with weight $1$.
Furthermore, we sample $s$ points uniformly at random (with replacement) $\xij \setminus Y_{i, j}$, and add to the set $\sij$ with weight of each point set to $\frac{|\xij \setminus Y_{ij} |}{s}$, which is an integer.
By construction, we have $|S_{i, j}| \le 2s$, and the weight of every point in $\sij$ is an integer. \fi
Finally, let $S = \bigcup_{i, j} \sij$. The set $S$ can be thought of as an $\epsilon$-coreset for the \kmedwo instance $\cI$. Even though we do not define this notion formally, the key properties of $S$ will be proven in Lemma \ref{lem:kmed-prob-bound} and \ref{lem:final-kmed-bound}. Thus, we will often informally refer to $S$ as a \emph{coreset}.

\begin{proposition} \label{prop:coreset-props}
	We have $|S| = O(((k+m)\log n /\epsilon)^2)$ if $\lambda$ is a constant.
	%\begin{itemize}
	%\item For any $c_i \in A$, $X_i = \bigcup_{j = 0}^{\phi} X_{i, j}$. 
	%\item The number of non-empty sets $X_{i, j}$ is at most $|A| \cdot (\phi+1) = O((k+m) \log n)$.
	%\item If $\lambda$ is a constant, then $|S| = O((k+m)(km)((\log n) /\epsilon)^2)$.
	%\end{itemize}
\end{proposition}
\begin{proof}
	For any $p \in X$, $\d(p, A) \le \deltam{0}(X, A) = \tau n \cdot R \le 2^{\phi} R$.
	Therefore, for any $c_i \in A$ and $j > \phi$, $X_{i, j'} = \emptyset$, and $X_i = \bigcup_{j = 0}^{\phi} X_{i, j}$.
	It follows that the number of non-empty sets $X_{i, j}$ is at most $|A| \cdot (1+\log(\tau n)) = O((k+m) \log n)$, since $|A| \le k+m$ and $\tau$ is a constant.
	For each non-empty $X_{i, j}$, $|S_{i, j}| \le 2s = O((m +  k\log n)/\epsilon^2)$, if $\lambda$ is a constant.
	Since $S = \bigcup_{i, j} S_{i, j}$, the claimed bound follows.
\end{proof}

\begin{proposition}{\cite{chen2009coresets,Haussler92}} \label{prop:concentration}
	Let $M \ge 0$ and $\eta$ be fixed constants, and let $h(\cdot)$ be a function defined on a set $V$ such that $\eta \le h(p) \le \eta + M$ for all $p \in V$. Let $U \subseteq V$ be a random sample of size $s$, and $\delta > 0$ be a parameter. If $s \ge \frac{M^2}{2\delta^2} \ln(2/\lambda)$, then
	\begin{equation*}
		\Pr \left[ \left| \frac{h(V)}{|V|} - \frac{h(U)}{|U|} \right| \ge \delta \right] \le \lambda,
	\end{equation*}
	%$\Pr \left[ \left| \frac{h(V)}{|V|} - \frac{h(U)}{|U|} \right| \ge \delta \right] \le \lambda$,
	where $h(U) \coloneqq \sum_{u \in U} h(u)$, and $h(V) \coloneqq \sum_{v \in V} h(v)$.
\end{proposition}

\begin{lemma} \label{lem:kmed-samplinglemma}
	Let $(\Gamma, \d)$ be a metric space, $V \subseteq \Gamma$ be a finite set of points, $\lambda', \xi > 0$, $q \ge 0$ be parameters, and define $s' = \frac{4}{\xi^2} \lr{q +  \ln\frac{2}{\lambda'}}$. 
	%If $|V| \ge s'$, and $U$ is a sample of $s'$ points picked uniformly and independently at random from $V$, with each point of $U$ having weight $|V|/|U|$, such that the total weight $w(U)$ is equal to $|V|$
	Suppose $U \subseteq V$ is a random sample of size $s'$.
	%We view $U$ as a weighted set where the weight of each point is $|V|/|U|$.
	Then for any fixed finite set $C \subseteq F$ with probability at least $1 - \lambda'$ it holds that for any $0 \le t \le q$, 
	\begin{equation*}
		\left|\deltam{t}(V,C) - \deltamw{ t'}(U, C) \right| \le \xi |V| (\diam(V) + \d(V, C)),
	\end{equation*}
	where $t' = \lfloor t|U|/|V| \rfloor$ and $w(u) = |V|/|U|$ for all $u \in U$.
\end{lemma}
\begin{proof}
	%We need sample of size $\left\lceil\max \LR{ \frac{4}{\xi^2} \cdot \ln(\frac{2}{\lambda'}),\allowbreak \frac{4q}{\xi}}\right\rceil \le s'$.
	Throughout the proof, we fix the set $C \subseteq F$ and $0 \le t \le q$ as in the statement of the lemma. Next, we define the following notation. For each $v \in V$, let $h(v) \coloneqq d(v, C)$, and let $h(V) \coloneqq \sum_{v \in V}h(v)$, and $h(U) \coloneqq \sum_{u \in U} h(u)$. Analogously, let $h'(V) \coloneqq \deltam{t}(V, C)$, and $h'(U) \coloneqq \deltam{t'}(U, C)$. Let $\eta(V) \coloneqq \min_{v \in V} \d(v, C)$, and $\eta(U) \coloneqq \min_{u \in U} \d(u, C)$. We summarize a few properties about these definitions in the following observation.
	
	\begin{observation} \label{obs:kmed-h}
		The following inequalities hold.
		\begin{itemize}
			\item $\lr{t\frac{|U|}{|V|}-1} \le t' \le t\frac{|U|}{|V|}$
			\item $h'(V) \le h(V) - t \cdot \eta(V) \le h(V)$, and $h'(V) \ge h(V) - t \cdot (\eta(V) + \diam(V))$
			\item $h'(U) \le h(U)$, and $h'(U) \ge h(U) - t\frac{|U|}{|V|} \cdot (\eta(U) + \diam(U))$
			\item $\eta(V) \le \eta(U) \le \eta(V) + \diam(V)$
		\end{itemize}
	\end{observation}
	\begin{proof}
		The first item is immediate from the definition $t' = \lfloor t|U|/|V| \rfloor$. Consider the second item. For each $v \in V$, let $g(v) \coloneqq \d(v, C) - \eta(V)$. Let $V' \subseteq V$ denote a set of points of size $t$ that have the $t$ largest distances to the centers in $C$. By triangle inequality, we get that $\d(v, C) \le \d(v, v^*) + \d(v^*, C) \le \diam(V) + \eta(V)$, where $v^* \in V$ is a point realizing the minimum distance $\eta(V)$ to the set of centers $C$. This implies that $g(v) \le \diam(V)$ for all $v \in V$. Now, observe that 
		\begin{align}
			h(V) &= h'(V) + \sum_{v \in V'} \lr{\eta(V) + g(v)} \nonumber \tag{Since $h'(V)$ excludes the distances of points in $V'$}
			\\&= h'(V) + t \cdot \eta(V) + \sum_{v \in V'} g(v) \label{eqn:kmed-ineq1}
			\\&\ge h'(V) + t \cdot \eta(V) \tag{$g(v) \ge 0$ for all $v \in V$}
		\end{align}
		By rearranging the last inequality, we get the first part of the second item. To see the second part, observe that (\ref{eqn:kmed-ineq1}) implies that $h(V) \le h'(V) + t\cdot \eta(V) + t\cdot\diam(V)$, since $g(v) \le \diam(V)$ for all $v \in V$.
		
		The proof of the third item is analogous to the proof of the first item. In addition, we need to combine the inequalities from the first item of the observation. We omit the details. The fourth item follows from the fact that $U \subseteq V$, and via triangle inequality.
	\end{proof}

	By applying Proposition~\ref{prop:concentration} with $\eta = \eta(V), M = \diam(V)$ and $\delta = \xi M/2$, we know with probability at most $\lambda'$,
	\begin{equation*}
		\left| \frac{\sum_{v \in V} \d(v, C)}{|V|} - \frac{\sum_{u \in U} \d(u, C) }{|U|} \right| \ge \frac{\xi}{2} \diam(V).
		%\\&\ \iff\ \left| \frac{h(V)}{|V|} - \frac{h(U)}{|U|} \right| \ge \frac{\xi}{2} \diam(V).
	\end{equation*}
	Recall that, $h(V) = \sum_{v \in V} \d(v, C)$ and $h(U) = \sum_{u \in U} \d(u, C)$. Thus, with probability at least $1-\lambda'$, we have that 
	\begin{equation}
		\left| \frac{h(V)}{|V|} - \frac{h(U)}{|U|} \right| \le \frac{\xi}{2} \diam(V). \label{eqn:kmed-ineq2}
	\end{equation}

	Now, we prove the following technical claim.
	\begin{claim} \label{clm-Delta}
		Suppose (\ref{eqn:kmed-ineq2}) holds. Then we have,
		\begin{equation}
			\left| \frac{h'(V)}{|V|} - \frac{h'(U)}{|U|} \right| \le \xi \cdot (\diam(V) + \d(V,C))  \label{eqn:kmed-ineq31}
		\end{equation}
	\end{claim}
	\begin{proof}
		We suppose that (\ref{eqn:kmed-ineq2}) holds, and show that (\ref{eqn:kmed-ineq31}) holds with probability 1.
		First, consider,
		\begin{align}
			\frac{h'(U)}{|U|} - \frac{h'(V)}{|V|} &\le \frac{h(U)}{|U|} - \frac{h(V)}{|V|} + \frac{t \cdot (\eta(V) + \diam(V))}{|V|} \nonumber \tag{From Observation~\ref{obs:kmed-h}, Part 2}
			\\&\le \frac{\xi}{2} \diam(V) + \frac{t}{|V|} \cdot (\eta(V)+\diam(V)) \tag{From (\ref{eqn:kmed-ineq2})} \nonumber
			\\&\le \frac{\xi}{2} \diam(V) + \frac{\xi}{2} \cdot (\eta(V)+\diam(V)) \label{eqn:kmed-ineq3}
		\end{align}
		where the last inequality follows from the assumption that $|V| \ge s' \ge \frac{4q}{\xi} \ge \frac{4t}{\xi}$.
		Now, consider
		\begingroup
		\allowdisplaybreaks
		\begin{align}
			\frac{h'(V)}{|V|} - \frac{h'(U)}{|U|} &\le \frac{h(V)}{|V|} - \frac{h(U)}{|U|} - \frac{t \eta(V)}{|V|} + \frac{t\frac{|U|}{|V|} \cdot (\eta(U) + \diam(U))}{|U|} \tag{From Observation~\ref{obs:kmed-h}, Part 3}
			\\&\le \frac{\xi}{2} \diam(V) - \frac{t \cdot  \eta(V)}{|V|} + \frac{t \cdot \eta(U)}{|V|} + \frac{t \cdot  \diam(U)}{|V|} \tag{From (\ref{eqn:kmed-ineq2})}
			\\&\le \frac{\xi}{2} \diam(V) - \frac{t \cdot \eta(V)}{|V|} + \frac{t \cdot (\eta(V) + \diam(V)) + t \cdot  \diam(U)}{|V|} \tag{From Observation~\ref{obs:kmed-h}, Part 4}
			\\&\le \frac{\xi}{2} \diam(V) + \frac{2t \cdot \diam(V)}{|V|} \tag{$\diam(U) \le \diam(V)$}
			\\&\le \xi \diam(V) \label{eqn:kmed-ineq4}
		\end{align}
		\endgroup
		where the last inequality follows from the assumption that $|V| \ge s' \ge \frac{4q}{\xi} \ge \frac{4t}{\xi}$.
	\end{proof}
	Thus, from Claim \ref{clm-Delta}, we know that since (\ref{eqn:kmed-ineq2}) holds with probability at least $1-\lambda'$, the following inequality also holds with probability at least $1-\lambda'$.
	\begin{equation*}
		\left| h'(V) - h'(U) \cdot \frac{|V|}{|U|} \right| \le \xi |V| \cdot (\diam(V) + \d(V, C)).
	\end{equation*} 
	The preceding inequality is equivalent to the one in the lemma, because $h'(V) = \deltam{t}(V, C)$, and $h'(U) \cdot \frac{|V|}{|U|} = \frac{|V|}{|U|} \cdot \deltam{t'}(U, C) = \deltamw{t'}(U, C)$.
	Finally, notice that Claim~\ref{clm-Delta} holds when the $h'$ function is defined with respect to any choice of $t \in \{0,1,\dots,q\}$.
	Therefore, with probability at least $1-\lambda'$, the inequality in the lemma holds for any $t \in \{0,1,\dots,q\}$, which completes the proof.
	%Note that since (\ref{eqn:kmed-ineq1}) is independent of the value of $t$, with probability at least $1-\lambda$, the inequality holds for each value $0 \le t \le q$. This concludes the proof of the lemma.
\end{proof}

Next, we show the following observation, whose proof is identical to an analogous proof in \cite{Chen08}.
\begin{observation} \label{obs:kmed-sum}
	The following inequalities hold.
	\begin{itemize}
		\item $\sum_{i, j} |X_{i, j}| 2^j R \le 3\cdot \deltam{0}(X, A) \le 3\tau \cdot \opt(\cI)$.
		\item $\sum_{i, j} |X_{i, j}| \diam(X_{i, j}) \le 6\cdot \deltam{0}(X, A) \le 6\tau \cdot \opt(\cI)$.
	\end{itemize}
\end{observation}
\begin{proof}
	%Fix an optimal solution $(X, Y^*, C^*)$, and let $O^* = X \setminus Y^*$ denote the set of $m$ outliers in this solution. For each $X_{i, j}$, let $Y^*_{i, j} = X_{i, j} \cap Y^*$, and let $t_{i, j} = |X_{i, j} \cap O^*|$. 
	
	For any $p \in X_{i, j}$, it holds that $2^j R \le \max \LR{2\d(p, A), R} \le 2\d(p, A) + R$. Therefore, 
	\begin{align*}
		\sum_{i, j} |X_{i, j}| \cdot 2^{j} R  &\le \sum_{i, j} \sum_{p \in X_{i, j}} 2^j R
		\\&\le \sum_{i, j} \sum_{p \in X_{i, j}} 2\d(p, A) + R
		\\&= 2 \sum_{p \in X} \d(p, A) + |X| \cdot |R|
		\\&= 2 \cdot \deltam{0}(X, A) + n|R|
		\\&\le 3 \cdot \deltam{0}(X, A) \tag{By definition of $R$}
		\\&\le 3\tau \opt(\cI') \le 3\tau \opt(\cI).
	\end{align*}
	We also obtain the second item by observing that $\diam(X_{i, j}) \le 2 \cdot 2^j \cdot R$.
\end{proof}

Next, we show that the following lemma, which informally states that the union of the sets of sampled points approximately preserve the cost of clustering w.r.t.\ \emph{any} set of at most $k$ centers, \emph{even after} excluding at most $m$ outliers overall. 

\begin{lemma} \label{lem:kmed-prob-bound}
	The following statement holds with probability at least $1-\lambda/2$:
	For all sets $C \subseteq F$ of size at most $k$, and for all sets of non-negative integers $\{m_{i, j}\}_{i, j}$ such that $\sum_{i, j} m_{i, j} \le m$,
	\begin{align}
		&\left| \sum_{i, j} \deltam{m_{i, j}}(X_{i, j}, C) - \sum_{i, j} \deltamw{t_{i, j}}(S_{i, j}, C) \right| \le \epsilon \cdot \sum_{i, j} \deltam{m_{i, j}}(X_{i, j}, C) \label{eqn:sum-bound}
	\end{align}
	where $\displaystyle t_{i, j} = \left\lfloor m_{i, j} / w_{i, j}\right\rfloor$.
\end{lemma}
\begin{proof}
	Fix an arbitrary set $C \subseteq F$ of at most $k$ centers, and the integers $\{m_{i,j}\}_{i, j}$ such that $\sum_{i, j} m_{i,j} \le m$ as in the statement of the lemma. For each $i = 1, \ldots, |A|$, and $0 \le j \le \phi$, we invoke Lemma \ref{lem:kmed-samplinglemma} by setting $V = \xij$, and $U = \sij$, $\xi = \frac{\epsilon}{8 \tau}$, $\lambda' = n^{-k} \lambda / (4(k+m)(1+\phi))$, and $q = m$. This implies that, the following inequality holds with probability at least $1-\lambda'$ for each set $\xij$, and the corresponding $m_{i, j} \le m$,
	\begin{align}
		&\left| \deltam{m_{i, j}}(X_{i, j},C) - \deltamw{t_{i,j}}(S_{i,j}, C) \right| \nonumber
		\\&\le \frac{\epsilon}{8\tau} |X_{i,j}| (\diam(X_{i,j}) + \d(X_{i, j}, C)) \label{eqn:kmed-sumineq}
	\end{align}
	
	Note that the sample size required in order for this inequality to hold is
	\begin{align*}
		s' &= \left\lceil \frac{4}{\xi^2} \lr{ m +  \ln \lr{\frac{2}{\lambda'}}} \right\rceil 
		\\&= \left\lceil 4\lr{\frac{8\tau}{\epsilon}}^2 \cdot \lr{ m + \ln\lr{\frac{8n^k(k+m)(1+\phi)}{\lambda}}}\right\rceil \le s.
	\end{align*}

	\iffalse Now, let $X' \subseteq X$ denote a subset of $m$ points that is farthest from the centers from $C$ (breaking ties arbitrarily) i.e., if the set of centers is fixed to be $C$, then $X'$ consists of $m$ outlier points. For each $i, j$, let $t_{i, j} = |X_{i, j} \cap X'|$, and $t'_{i, j} = \lfloor t_{i, j}|U|/|V|\rfloor$. Note that for any $i, j$, $0 \le t'_{i, j} \le t_{i, j} \le m$, and $\sum_{i, j} t'_{i, j} \le \sum_{i, j} t_{i, j} = m$. Then, note that $\Delta(C) = \sum_{i, j} \deltamw{t'_{i, j}}(S_{i, j}, C)$ satisfies all the conditions as stated in the statement of the lemma. Therefore, 
	\begin{align}
		\left| \deltam{m}(X, C) - \Delta(C) \right| &\le \sum_{i, j} \left| \deltam{t_{i, j}}(X_{i, j}, C) - \deltamw{t'_{i, j}}(S_{i, j}, C) \right| \label{eqn:kmed-sumineq2}
	\end{align}
	\fi
	
	For any $i, j$, if $\xij < s$ (i.e., $\xij$ is \emph{small}), then the sample $S_{i,j}$ is equal to $X_{i, j}$, and each point in $S_{i, j}$ has weight equal to $1$.
	This implies that $\deltam{m_{i, j}}(X_{i, j}, C) = \deltamw{t_{i, j}}(S_{i, j}, C)$ for all such $X_{i, j}$, and their contribution to the right hand side of inequality (\ref{eqn:sum-bound}) is zero. Thus, it suffices to restrict the sum on the right hand side of (\ref{eqn:sum-bound}) over \emph{large} sets $X_{i, j}$'s.
	Let $\mathcal{L}$ consist of all pairs $(i,j)$ such that $X_{i,j}$ is large.
	We have the following claim.
	\begin{claim} \label{cl:kmed-sum-large}
		$\sum_{(i, j) \in \mathcal{L}} |X_{i, j}|\d(X_{i, j}, C) \le 2\deltam{m}(X, C).$
	\end{claim}
	\begin{proof}
		Let $Y$ denote the farthest $m$ points in $X$ from the set of centers $C$. Now, fix $(i, j) \in \mathcal{L}$ and let $q_{i, j} \coloneqq |X_{i, j} \cap Y| \le m$ denote the number of outliers in $X_{i, j}$. Since $|X_{i, j}| \ge 2m \ge 2q_{i, j}$, the set $X_{i, j} \setminus Y$ is non-empty, and all points $X_{i, j} \setminus Y$ contribute towards $\deltam{m}(X, C)$. That is, 
		\begin{equation}
			\sum_{(i, j) \in \mathcal{L}} \sum_{p \in X_{i, j} \setminus Y} \d(p, C) \le \deltam{m}(X, C) \label{eqn:kmed-summinus}
		\end{equation}
		
		For any $p \in X_{i, j} \setminus Y$, $\d(X_{i, j}, C) \le \d(p, C)$ from the definition. Therefore, 
		\begin{align*}
			&\sum_{(i, j) \in \mathcal{L}} |X_{i, j}| \cdot \d(X_{i, j}, C) 
			\\&\le \sum_{(i, j) \in \mathcal{L}} 2|X_{i, j} \setminus Y| \cdot \d(X_{i, j}, C) 
			\\&\le 2 \cdot \sum_{(i, j) \in \mathcal{L}}\ \sum_{p \in X_{i, j} \setminus Y} \d(p, C)
			\\&\le 2 \cdot \deltam{m}(X, C) 
		\end{align*}
		Here, to see the second inequality, see that $|X_{i, j}| \ge 2q_{i, j}$, which implies that $|X_{i, j}|- q_{i, j} \le 2(|X_{i, j}| - q_{i, j})$. The last inequality follows from (\ref{eqn:kmed-summinus}).
	\end{proof}
	Thus, by revisiting (\ref{eqn:sum-bound}) and (\ref{eqn:kmed-sumineq}), we get:
	\begin{align*}
		&\sum_{(i, j) \in \mathcal{L}} \left| \deltam{m_{i, j}}(X_{i, j}, C) - \deltamw{t_{i, j}}(S_{i, j}, C) \right| 
		\\&\le \frac{\epsilon}{8\tau} \sum_{(i, j) \in \mathcal{L}} |X_{i,j}| (\diam(X_{i,j}) + \d(X_{i, j}, C)) \tag{From (\ref{eqn:kmed-sumineq})}
		\\&\le \frac{\epsilon}{8\tau} \cdot (6 \tau \cdot \opt(\cI) + 2\deltam{m}(X, C)) \tag{From Obs. \ref{obs:kmed-sum} and Claim \ref{cl:kmed-sum-large}}
		\\&= \frac{\epsilon}{8\tau} (8\tau \cdot \deltam{m}(X, C)) = \epsilon \cdot \deltam{m}(X, C)
	\end{align*}
	Where, in the last inequality, since $C$ is an arbitrary set of at most $k$ centers, $\opt(\cI) \le \deltam{m}(X, C)$.
	Note that the preceding inequality holds for a fixed set $C$ of centers with probability at least $1 - |A|\cdot (1+\phi)\lambda' = 1-n^{-k}\lambda/2$, which follows from taking the union bound over all sets $X_{i, j}$, $1 \le i \le |A| \le k+m$, and $0 \le j \le \phi$. 
	
	Since $F$ has at most $n^k$ subsets of size at most $k$, the statement of the lemma follows from taking a union bound.
	\iffalse 
	\begin{observation}
		For any set $P$ with $|P| > s \ge 2m$, the following holds:
		\begin{enumerate}
			\item For any $0 \le m' \le m$, \ \ $|P|/2 \le |P| - m \le |P|$
			\item For any $0 \le m_1, m_2 \le m$,\ \  $|P|-m_1 \le |P| \le 2(|P|-m_2)$.
		\end{enumerate}
	\end{observation}
	\fi
\end{proof}

%\todo{Jie: I will add a paragraph here for proving the main theorem.}

Now we are ready to prove Theorem~\ref{thm-generalreduction}.
%Define $S = \bigcup_{i,j} S_{i,j}$.
We enumerate every subset $T \subseteq S$ of size at most $m$.
For each $T$, we compute a $\beta$-approximation solution for the (weighted) $k$-median instance $(S \backslash T,F,k)$.
Theorem~\ref{thm-generalreduction} only assumes the existence of a $\beta$-approximation algorithm for unweighted $k$-median, which cannot be applied to weighted point sets.
However, we can transform $S \backslash T$ to an equivalent unweighted sets $R$, which contains, for each $x \in S \backslash T$, $w(x)$ copies of (unweighted) $x$, where $w(x)$ is the weight of $x$ in $S \backslash T$.
It is clear that $\mathsf{wcost}(S \backslash T,C) = \mathsf{cost}(R,C)$ for all $C \subseteq F$.
Thus, we can apply the $\beta$-approximation \kmed algorithm on $(R,F,k)$ to compute a center set $C \subseteq F$ of size $k$ such that $\mathsf{wcost}(T,C) \leq \beta \cdot \mathsf{wcost}(T,C')$ for any $C' \subseteq F$ of size $k$.
We do this for all $T \subseteq S$ of size at most $m$.
Let $\mathcal{C}$ denote the set of all center sets $C$ computed.
We pick a center set $C^* \subseteq \mathcal{C}$ that minimizes $\mathsf{cost}_m(X,C^*)$, and return $(Y^*,C^*)$ as the solution where $Y^* \subseteq X$ consists of the $m$ points in $X$ farthest to the center set $C^*$.

\begin{lemma} \label{lem:final-kmed-bound}
	With probability at least $1- \frac{\lambda}{2}$, for all $C \subseteq F$ of size $k$ we have
	\begin{equation*}
		\mathsf{cost}_m(X,C^*) \leq \frac{1+\epsilon}{1-\epsilon} \cdot \beta \mathsf{cost}_m(X,C).
	\end{equation*}
\end{lemma}
\begin{proof}
	The statement in Lemma~\ref{lem:kmed-prob-bound} holds with probability at least $1- \lambda/2$.
	Thus, it suffices to assume the statement in Lemma~\ref{lem:kmed-prob-bound}, and show $\mathsf{cost}_m(X,C^*) \leq (1+\epsilon)^2 \beta \cdot \mathsf{cost}_m(X,C)$ for any $C \subseteq F$ of size $k$.
	Fix a subset $C \subseteq F$ of size $k$.
	Let $Y \subseteq X$ consist of the $m$ points in $X$ farthest to $C$, and define $m_{i,j} = |Y \cap X_{i,j}|$.
	Set $t_{i,j} = \lfloor m_{i,j}/w_{i,j} \rfloor$.
	Note that $\mathsf{cost}_m(X,C) = \sum_{i,j} \mathsf{cost}_{m_{i,j}}(X_{i,j},C)$.
	Furthermore, by Lemma~\ref{lem:kmed-prob-bound}, we have
	\begin{align} 
		\sum_{i,j} \mathsf{wcost}_{t_{i,j}}(S_{i,j},C) &\leq (1+\epsilon) \cdot \sum_{i,j} \mathsf{cost}_{m_{i,j}}(X_{i,j},C)\nonumber 
		\\&= (1+\epsilon) \cdot \mathsf{cost}_m(X,C). \label{eq-CtoC}
	\end{align}
	Now let $T_{i,j} \subseteq S_{i,j}$ consist of the $t_{i,j}$ points in $S_{i,j}$ farthest to $C$, and define $T = \bigcup_{i,j} T_{i,j}$.
	Since $|T| \leq m$, $T$ is considered by our algorithm and thus there exists a center set $C' \in \mathcal{C}$ that is a $\beta$-approximation solution for the (weighted) $k$-median instance $(S \backslash T, F, k)$.
	We have 
	\begin{align} 
		\mathsf{wcost}(S \backslash T,C') &\leq \beta \cdot \mathsf{wcost}(S \backslash T,C) = \beta \sum_{i,j} \mathsf{wcost}_{t_{i,j}}(S_{i,j},C). \label{eq-C'toC}
	\end{align}
	Note that $\mathsf{wcost}(S \backslash T,C') \geq \sum_{i,j} \mathsf{wcost}_{t_{i,j}}(S_{i,j},C')$.
	Furthermore, by applying Lemma~\ref{lem:kmed-prob-bound} again, we have $\sum_{i,j} \mathsf{wcost}_{t_{i,j}}(S_{i,j},C') \geq (1-\epsilon) \cdot \sum_{i,j} \mathsf{cost}_{m_{i,j}}(X_{i,j},C')$.
	It then follows that
	\begin{align} 
		(1-\epsilon) \cdot \mathsf{cost}_m(X_{i,j},C') &\leq (1-\epsilon) \cdot \sum_{i,j} \mathsf{cost}_{m_{i,j}}(X_{i,j},C') \leq \mathsf{wcost}(S \backslash T,C'). \label{eq-C'toC'}
	\end{align}
	Finally, we have $\mathsf{cost}_m(X,C^*) \leq \mathsf{cost}_m(X,C')$ by the construction of $C^*$.
	Combining this with (\ref{eq-CtoC}), (\ref{eq-C'toC}), and (\ref{eq-C'toC'}), we have $\mathsf{cost}_m(X,C^*) \leq \frac{1+\epsilon}{1-\epsilon} \cdot \beta \mathsf{cost}_m(X,C)$, which completes the proof.
	%By the construction of $C^*$, we have $\mathsf{cost}_m(X,C^*) \leq \mathsf{cost}_m(X,C')$.
	%So it suffices to show $\mathsf{cost}_m(X,C') \leq \frac{1+\epsilon}{1-\epsilon} \cdot \beta \mathsf{cost}_m(X,C)$.
\end{proof}

By choosing $\lambda>0$ to be a sufficiently small constant, and by appropriately rescaling $\epsilon$ \footnote{Since Lemma \ref{lem:final-kmed-bound} implies a $\beta(1+O(\epsilon))$-approximation, and $\beta$ is a constant, it suffices to redefine $\epsilon = \epsilon/c$ for some large enough constant $c$ to get the desired result.}, the above lemma shows that our algorithm outputs a $(\beta+\epsilon)$-approximation solution with a constant probability. By repeating the algorithm a logarithmic number of rounds, we can guarantee the algorithm succeeds with high probability.
%From Proposition \ref{prop:coreset-props}, we get that $|S| = O((mk \log n/\epsilon)^2)$.
The number of subsets $T \subseteq S$ of size at most $m$ is bounded by $|S|^{O(m)}$, which is $\lr{\frac{(k+m) \log n}{\epsilon}}^{O(m)}$ by Proposition \ref{prop:coreset-props}.
Note that $(\log n)^{O(m)} \leq \max\{m^{O(m)}, n^{O(1)}\}$.
Thus, the number of subsets $T \subseteq S$ of size at most $m$ is bounded by $f(k, m, \epsilon) \cdot n^{O(1)}$, where $f(k, m, \epsilon) = \lr{\frac{k+m}{\epsilon}}^{O(m)}$.
Thus, we need to call the $\beta$-approximation \kmed algorithm $f(k, m, \epsilon) \cdot n^{O(1)}$ times, which takes $f(k,m,\epsilon) n^{O(1)} \cdot T(n,k)$ time overall.
The first call of the algorithm for obtaining a $\tau$-approximation to the $(k+m)$-\textsc{Median} instance takes polynomial time. Besides this, the other parts of our algorithm can all be done in polynomial time. This completes the proof of Theorem~\ref{thm-generalreduction}.

\subsection{Ensuring Integral Weights in the Coreset} \label{subsec:intweights}
Recall that in order to obtain the set $S_{i, j}$ from a \emph{large} $X_{i, j}$, we sample $s$ points uniformly and independently at random (with replacement), and give each point in $S_{i, j}$ the weight $w_{i,j} = \frac{|X_{i, j}|}{|S_{i, j}|}$. In the main body of the proof, we assumed that the quantity $w_{i, j}$ is integral for the sake of simplicity. However, in general $\frac{|X_{i, j}|}{|S_{i, j}|}$ may not be an integer. Here, we describe how to modify this construction to ensure integral weights.

To this end, let $X^{(1)}_{i, j} \subseteq X_{i, j}$ be an arbitrary subset of size $|\xij| \mod s$, and let $X^{(2)}_{i, j} = X_{i, j} \setminus Y_{i, j}$. From this time onward, we \emph{treat} $X^{(1)}_{i, j}$ and $X^{(2)}_{i, j}$ as two separate sets of the form $X_{\cdot, \cdot}$, and proceed with the construction of the coreset. 

In particular, observe that $|X^{(1)}_{i, j}| < s$, i.e., it is \emph{small}, and $|X^{(2)}_{i, j}| = t \cdot s$ for some positive integer $t$, and thus $X^{(2)}_{i, j}$ is \emph{large}. Therefore, we let $S^{(1)}_{i, j} \gets X^{(1)}_{i, j}$, and each point is added with weight $1$. On the other hand, to obtain $S^{(2)}_{i, j}$, we sample $s$ points uniformly and independently at random from $X^{(2)}_{i, j}$, and set the weight of each point to be $|X^{(2)}_{i, j}|/s$, which is an integer. From this point onward, we proceed with exactly the same analysis as in the original proof, i.e., we treat $X^{(1)}_{i, j}$ as a \emph{small} set, and $X^{(2)}_{i,j}$ as a large set in the analysis. Since for the small sets, the sampled set is equal to the original set, their contribution to the left hand side of the following inequality in the statement of Lemma \ref{lem:kmed-prob-bound}, is equal to zero. 
\begin{align*}
	&\left| \sum_{i, j} \deltam{m_{i, j}}(X_{i, j}, C) - \sum_{i, j} \deltamw{t_{i, j}}(S_{i, j}, C) \right| \le \epsilon \cdot \sum_{i, j} \deltam{m_{i, j}}(X_{i, j}, C)
\end{align*}
Therefore, the analysis of Lemma \ref{lem:kmed-prob-bound} goes through without any modifications. The only other minor change is that the number of points in the coreset $S$, which is obtained by taking the union of all $S_{\cdot, \cdot}$, is now at most twice the previous bound, which is easily absorbed in the big-oh notation.

%!TEX root = mainpaper.tex

\section{Extensions} \label{sec:extensions}

\subsection{$k$-Means with Outliers}
This is similar to \kmedwo, except that the cost function is the sum of \emph{squares} of distances of all except $m$ outlier points to a set of $k$ facilities. This generalizes the well-known \kmeans problem. Here, the main obstacle is that, the squares of distances do not satisfy triangle inequality, and thus it does not form a metric. However, they satisfy a \emph{relaxed} version of triangle inequality (i.e., $\d(p, q)^2 \le 2 (\d(p, r)^2 + \d(r, q)^2)$). This technicality makes the arguments tedious, nevertheless, we can follow the same approach as for \kmedwo, to obtain optimal FPT approximation schemes. Our technique implies an optimal $(1+\frac{8}{e} + \epsilon)$-approximation for \kmeanswo (using the result of \cite{DBLP:conf/icalp/Cohen-AddadG0LL19} as a black-box), improving upon polynomial-time $53.002$-approximation from \cite{KrishnaswamyLS18}, and $(9+\epsilon)$-approximation from \cite{GoyalJ020} in time FPT in $k, m$ and $\epsilon$.

In fact, using our technique, we can get improved approximation guarantees for \kzwo, where the cost function involves $z$-th power of distances, where $z \ge 1$ is fixed for a problem. Note that the cases $z = 1$ and $z= 2$ correspond to \kmedwo and \kmeanswo respectively. We give the details for \kzwo in the appendix.

\subsection{Matroid Median with Outliers} A \emph{matroid} is a pair $\mathcal{M} = (F, \mathcal{S})$, where $F$ is a ground set, and $\mathcal{S}$ is a collection of subsets of $F$ with the following properties: (i) $\emptyset \in \mathcal{S}$, (ii) If $A \in \mathcal{S}$, then for every subset $B \subseteq A$, $B \in \mathcal{S}$, and (iii) For any $A, B \in \mathcal{S}$ with $|B| < |A|$, there exists an $b \in B \setminus A$ such that $B \cup \{b\} \in \mathcal{S}$. The rank of a matroid $\mathcal{M}$ is the size of the largest independent set in $\mathcal{S}$. Using the definition of matroid, it can be easily seen that all inclusion-wise maximal independent sets (called \emph{bases}) have the same size.

An instance of \matmedwo is given by $(X, F, \mathcal{M}, m)$, where $\mathcal{M} = (F, \mathcal{S})$ is a matroid with rank $k$ defined over a finite ground set $F$, and $X, F$ are sets of clients and facilities, belonging to a finite metric space $(\Gamma, \d)$. The objective is to find a set $C \subseteq F$ of facilities that minimizes $\deltam{m}(X, C)$, and $C \in \mathcal{S}$, i.e., $C$ is an independent set in the given matroid. Note that an explicit description of a matroid of rank $k$ may be as large as $n^k$. Therefore, we assume that we are given an efficient \emph{oracle access} to the matroid $\mathcal{M}$. That is, we are provided with an algorithm $\mathcal{A}$ that, given a candidate set $S \subseteq F$, returns in time $T(\mathcal{A})$ (which is assumed to be polynomial in $|F|$), returns whether $S \in \mathcal{I}$. 

We can adapt our approach to \matmedwo in a relatively straightforward manner. Recall that our algorithm needs to start with an instance of outlier-free problem (i.e., \matmed) that provides a lower bound on the optimal cost of the given instance. To this end, given an instance $\cI= (X, F, \mathcal{M} = (F, \mathcal{S}), m)$ of \matmedwo, we define an instance $\cI' = (X, F, \mathcal{M'}, 0)$ of Matroid Median with $0$ Outliers (i.e., Matroid Median), where $\mathcal{M}' = (F \cup X, \mathcal{S}')$ is defined as follows. $\mathcal{S}' = \{Y \cup C : Y \subseteq X \text{ with } |Y| \le m \text{ and } C \subseteq F \text{ with } C \in \mathcal{S}\}$. That is, each independent set of $\mathcal{M}'$ is obtained by taking the union of an independent set of facilities from $\mathcal{M}$, and a subset of $X$ of size at most $m$. It is straightforward to show that $\mathcal{M}'$ satisfies all three axioms mentioned above, and thus is a matroid over the ground set $F \cup X$. In particular, it is the direct sum of $\mathcal{M}$ and a uniform matroid $\mathcal{M}_m$ over $X$ of rank $m$ (i.e., where any subset of $X$ of size at most $m$ is independent). Note that using the oracle algorithm $\mathcal{A}$, we can simulate an oracle algorithm to test whether a candidate set $C \subseteq F \cup X$ is independent in $\mathcal{M}'$. Therefore, using a $(2+\epsilon)$-approximation for \matmed \cite{DBLP:conf/icalp/Cohen-AddadG0LL19} in time FPT in $k$ and $\epsilon$, we can find a set $A \subseteq F \cup X$ of size at most $k+m$ that we can use to construct a coreset. The details about enumeration are similar to that for \kmedwo, and are thus omitted.

\subsection{Colorful $k$-Median} 

This is an orthogonal generalization of \kmedwo to ensure a certain notion of \emph{fairness} in the solution (see \cite{JiaSS20}). Suppose the set of points $X$ is partitioned into $\ell$ different colors $X_1 \uplus X_2 \uplus \ldots \uplus X_\ell$. We are also given the corresponding number of outliers $m_1, m_2, \ldots, m_\ell$. The goal is to find a set of at most facilities $C$ to minimize the connection cost of all except at most $m_t$ outliers from each color class $X_t$, i.e., we want to minimize the cost function: $\sum_{t = 1}^\ell \deltam{m_t}(X_t, C)$. This follows a generalizations of the well-known $k$-\textsc{Center} problem introduced in \cite{Bandyapadhyay0P19} and \cite{AneggAKZ20,JiaSS20} , called \textsc{Colorful $k$-Center}. Similar generalization of \textsc{Facility Location} has also been studied in \cite{chekuri2022algorithms}. 

Using our ideas, we can find an FPT approximation parameterized by $k$, $m = \sum_{t = 1}^\ell m_t$, and $\epsilon$. To this end, we sample sufficiently many points from each color class $X_t$ separately, and argue that it preserves the cost appropriately. The technical details follow the same outline as that for $k$-Median with $m$ Outliers. In particular, during the enumeration phase---just like that for \kmedwo---we obtain several instances of \kmed. That is, our algorithm is \emph{color-agnostic} after constructing the coreset. Thus, we obtain a tight $(1+\frac{2}{e}+\epsilon)$-approximation for this problem. This is the first non-trivial true approximation for this problem -- previous work \cite{GuptaMZ21} only gives a \emph{pseudo-approximation}, i.e., a solution with cost at most a constant times that of an optimal cost, but using slightly more than $k$ facilities. 

\subsection{A Combination of Above Generalizations}

Our technique also works for a combination of the aforementioned generalizations that are orthogonal to each other. To consider an extreme example, consider \textsc{Colorful Matroid Median} with $\ell$ different color classes (a similar version for $k$-\textsc{Center} objective has been recently studied by \cite{AneggKZ22}), where we want to find a set of facilities that is independent in the given matroid, in order to minimize the sum of distances of all except $m_t$ outlier points for each color class $X_t$. By using a combination of the ideas mentioned above, one can get FPT approximations for such generalizations. 

%!TEX root = mainpaper.tex

\section{Concluding Remarks} 
In this paper, we give a reduction from \kmedwo to \kmed that runs in time FPT in $k, m$, and $\epsilon$, and preserves the approximation ratio up to an additive $\epsilon$ factor. As a consequence, we obtain improved FPT approximations for \kmedwo in general as well as special kinds of metrics, and these approximation guarantees are known to be tight in general. Furthermore, our technique is versatile in that it also gives improved approximations for related clustering problems, such as \kmeanswo, \matmedwo, and \textsc{Colorful $k$-Median}, among others. 

The most natural direction is to improve the FPT running time while obtaining the tight approximation ratios. More fundamentally, perhaps, is the question whether we need an FPT dependence on the number of outliers, $m$; or whether it is possible to obtain approximation guarantees for \kmedwo matching that for \kmed, with a running time that is FPT in $k$ and $\epsilon$ alone.

\subsection*{Acknowledgments.}
We thank Rajni Dabas for bringing a few minor typos to our attention.
\\T. Inamdar is supported by  the  European  Research  Council  (ERC)  under  the  European  Union’s  Horizon  2020  research  and innovation programme (grant agreement No 819416). S. Saurabh is supported by the European  Research  Council  (ERC)  under  the  European  Union’s  Horizon  2020  research  and innovation programme (grant agreement No 819416) and Swarnajayanti Fellowship (No DST/SJF/MSA01/2017-18). 

\bibliographystyle{siam}
\bibliography{references.bib}

\appendix
%!TEX root = mainpaper.tex

\section{$(k,z)$-clustering with Outliers}
Let $z \ge 1$ be a fixed real that is not part of the input of the problem.
\\The input of the $(k, z)$-\textsc{Clustering} problem is an instance $\cI = ((\Gamma, \d), X, F, k)$, where $(\Gamma, \d)$ is a metric space, $X \subseteq \Gamma$ is a (finite) set of $n$ points, called \emph{points} or \emph{clients}, $F \subseteq \Gamma$ is a set of \emph{facilities}, and $k$ is a positive integer. 
The task is to find a set $C \subseteq F$ of facilities (called \emph{centers}) in $F$ that minimizes the following cost function:
\[\cost(X, C) \coloneqq \sum_{p \in X} \cost(p, C) \]
where $\cost(p, C) \coloneqq (\d(p, C))^z$. %Since we assume that $z$ is fixed for the entire discussion, we omit the subscript $z$, and use the notation $\cost_z(\cdot, \cdot)$.

%Note that in some cases, the metric space $(\Gamma, \d)$, and the set $F$ of facilities may be infinite, e.g., when $V = \real^d$ some $d \ge 1$. In this case, it is natural to assume that the distance $\d(p, q)$ for any two points $p, q \in \Gamma$ can be efficiently computed. For example, when $\Gamma = \real^d$, it is typical to assume that the distance $\d(x, y)$ between a pair of points $p, q \in \Gamma$, is given by $\ell_p$ distance, for some $1 \le p \le \infty$, with $p = 2$---which corresponds to the Euclidean distance---being a particularly interesting case.

\paragraph{$(k, z)$-clustering with $m$ outliers.}
Here, the input contains an additional integer $1 \le m \le n$, and the goal is to find a set $X' \subseteq X$ of $n-m$ points, such that $\cost(X', C)$ is minimized (over choices all of $X'$ and $C$). Here, the set $X \setminus X'$ of at most $m$ points corresponds to the set of \emph{outliers}. In another notation, we want to find a set $C \subseteq F$ of at most $k$ centers that minimizes $\deltam{m}(X, C) \coloneqq \summ{m} \{\cost(p, C) : p \in X \}$, i.e., the sum of $n-m$ smallest distances of points in $X$ to the set of centers $C$.

First, we state a few properties about the $z$-th powers of distances, which will be subsequently useful in the analysis.
\begin{proposition} \label{prop:distance-z}
	Let $P, C \subseteq \Gamma$ be non-empty finite subsets of points. For any point $p \in P$, the following holds:
	\begin{itemize}
		\item $\d(P, C)^z \le \cost(p, C)  \le (\d(P, C) + \diam(P))^z \le 2^{z} \cdot \lr{\d(P, C)^z + \diam(P)^z}$
		\item $(\d(p, C) - \d(P, C))^z \le (\diam(P))^z$
	\end{itemize}
\end{proposition}
\begin{proof}
	Let $p^* \in P$ be a point realizing the smallest distance $\d(P, C)$. It follows that for any $p \in P$, 
	\begin{align*}
		\d(P, C) = \d(p^*, C) &\le \d(p, C) 
		\\&\le \d(p, p^*) + \d(p^*, C) \tag{by triangle inequality} 
		\\&\le \diam(P) + \d(P, C)  \tag{$\d(p, p^*) \le \diam(P)$}
		\\&\le 2 \max\{\diam(P), \d(P, C)\}
	\end{align*}
	Now, by taking the $z$-th power of each term, we get the first inequality, which follows from $\max\{a, b\} \le a+b$.
	
	Note that the first and third line in the preceding chain of inequalities implies that $(\d(p, C) - \d(P, C)) \le \diam(P)$. Note that both sides of the inequality are non-negative. Thus, by taking the $z$-th power of both sides, the second item follows.
\end{proof}

Consider an instance $\cI = ((\Gamma, \d), X, F, k, m)$ be an instance of $(k, z)$-\textsc{Clustering with $m$ Outliers}. We define an instance $\cI' = ((\Gamma, \d), X, F \cup X, k+m, 0)$ of $(k+m, z)$-\textsc{Clustering} (without outliers), where in addition to the original set of facilities, there is a facility co-located with each client. The following observation and its proof is analogous to Observation \ref{obs:kmed-lb}, and thus we omit the proof.

\begin{observation}\label{obs:kz-lb}
	$\opt(\cI') \le \opt(\cI)$, i.e., the value of an optimal solution to $\cI'$ is a lower bound on the value of an optimal solution to $\cI$.
\end{observation}

The following definitions and the construction of the coreset is analogous to that for $k$-median with $m$ outliers, with appropriate modifications needed for $z$-th power of distances. First, we assume that there exists a $\tau$-approximation algorithm for $(k, z)$-clustering problem that runs in polynomial time, where $\tau = O(1)$. Then, by using this $\tau$-approximation algorithm for the instance $\cI'$, we obtain a set of at most $k' \le k+m$ centers $A$ such that $\deltam{0}(X, A) \le \tau \cdot \opt(\cI') \le \tau \cdot \opt(\cI)$. Let $R = \lr{\frac{\deltam{0}(X, A)}{\tau n}}^{1/z}$ be a lower bound on average radius, and let $\phi = \lceil \log(\tau n) \rceil$. For each $c_i \in A$, let $X_i \subseteq X$ denote the set of points whose closest center in $A$ is $c_i$. By arbitrarily breaking ties, we assume that the sets $X_i$ are disjoint, i.e., the sets $\{X_i\}_{1 \le i \le k'}$ form a partition of $X$. Now, we define the set of \emph{rings} centered at each center $c_i$ as follows.
\[\xij \coloneqq \begin{cases}
	B_{X_i}(c_i, R) & \text{ if } j = 0
	\\B_{X_i}(c_i, 2^j R) \setminus  B_{X_i}(c_i, 2^{j-1} R) & \text{ if } j \ge 1
\end{cases}
\]

Let $s = \frac{c \tau^2 2^{(c'z)}}{\epsilon^2} \lr{m + k \ln n  + \ln(1/\lambda)}$, for some large enough constants $c, c'$. We define a weighted set of points $\sij \subseteq \xij$ as follows. If $|\xij| \le s$, then say that $\xij$ is \emph{small}, and let $\sij \coloneqq \xij$, and let the weight of each point $p \in \sij$ be $1$. Otherwise, if $|\xij| > s$, then say that $\xij$ is \emph{large}. In this case, let $Y_{i, j} \subseteq X_{i, j}$ be an arbitrary subset of size $|\xij| \mod q$. We add each point $q \in Y_{i, j}$ to $\sij$ with weight $1$. Furthermore, we sample $s$ points uniformly at random (with replacement) $\xij \setminus Y_{i, j}$, and add to the set $\sij$ with weight equal to $\frac{|\xij \setminus Y_{ij} |}{s}$, which is an integer. Thus, we assume that $|S_{i, j}| \le 2s$, and the weight of every point in $\sij$ is an integer. Finally, let $S = \bigcup_{i, j} \sij$.

\begin{lemma} \label{lem:kz-samplinglemma}
	Let $(\Gamma, \d)$ be a metric space, and let $V \subseteq \Gamma$ be a finite set of points. Let $\lambda', \xi > 0$, $q \ge 0$, be parameters, and define $s' = \frac{4}{\xi^2} \lr{q + \ln\frac{2}{\lambda'}}$. If $|V| \ge s'$, and $U$ is a sample of $s'$ points picked uniformly and independently at random from $V$, with each point of $U$ having weight $|V|/|U|$, such that the total weight $w(U)$ is equal to $|V|$, then for any fixed finite set $C \subseteq \Gamma$, and for any $0 \le t \le q$, with probability at least $1 - \lambda'$ it holds that
	\begin{align}
		&\left|\deltam{t}(V,C) - \deltamw{ t'}(U, C) \right| \le 2^{2z+2} \xi |V| \cdot (\diam(V)^z + \d(V, C)^z), \label{eqn:kz-mainineq}
	\end{align}
	where $t' = \lfloor t|U|/|V| \rfloor$.
\end{lemma}
\begin{proof}
	%We need sample of size $\left\lceil\max \LR{ \frac{4}{\xi^2} \cdot \ln(\frac{2}{\lambda'}),\allowbreak \frac{4q}{\xi}}\right\rceil \le s'$.
	Throughout the proof, we fix the set $C$ and $0 \le t \le q$ as in the statement of the lemma. Next, we define the following notation. For all $v \in V$, let $h(v) = \cost(v, C) = \d(v, C)^z$, and let $h(V) \coloneqq \sum_{v \in V}h(v)$, and $h(U) \coloneqq \sum_{u \in U} h(u)$. Analogously, let $h'(V) \coloneqq \deltam{t}(V, C)$, and $h'(U) \coloneqq \deltam{t'}(U, C)$, i.e., sum of all except $t$ (resp.\ $t'$) largest $h$-values. Let $\eta(V) \coloneqq \min_{v \in V} \d(v, C)^z$, and $\eta(U) \coloneqq \min_{u \in U} \d(u, C)$. We summarize a few properties about these definitions in the following observation, which is analogous to Observation \ref{obs:kmed-h}.
	\begin{observation} \label{obs:kz-h} \
		\begin{itemize}
			\item $\lr{t\frac{|U|}{|V|}-1} \le t' \le t\frac{|U|}{|V|}$
			\item For any $p \in P$, $(\eta(V)^z \le \d(p, C)^z = \cost(p, C)^z \le \eta(V)^z + 2^z (\eta(V)^z + \diam(V))$
			\item $h'(V) \le h(V) - t \cdot \eta(V)^z \le h(V)$, and $h'(V) \ge h(V) - 2^z \cdot t \cdot (\eta(V)^z + \diam(V)^z)$
			\item $h'(U) \le h(U)$, and $h'(U) \ge h(U) - 2^z \cdot t \frac{|U|}{|V|} \cdot (\eta(U)^z + \diam(U)^z)$
			\item $\eta(V)^z \le \eta(U)^z \le 2^z (\eta(V)^z + \diam(V)^z)$
		\end{itemize}
	\end{observation}
	\begin{proof}
		The first item is immediate from the definition $t' = \lfloor t|U|/|V| \rfloor$. 
		
		Consider the second item. For each $v \in V$, let $g(v) \coloneqq \d(v, C) - \eta(V)$. Let $V' \subseteq V$ denote a set of points of size $t$ that have the $t$ largest distances to the centers in $C$. From Proposition \ref{prop:distance-z}, we get that for any $p \in V$, $\eta(V)^z \le \cost(p, C)^z \le 2^z \cdot (\eta(V)^z + \diam(P)^z)$. This implies that $g(v) \le \diam(V)$ for all $v \in V$. Now, observe that 
		\begin{align}
			h(V) &= h'(V) + \sum_{v \in V'} \lr{\eta(V) + g(v)}^z \nonumber \tag{Since $h'(V)$ excludes the distances of points in $V'$}
			\\&\ge h'(V) + t \cdot \eta(V)^z \tag{$g(v) \ge 0$ for all $v \in V$}
		\end{align}
		By rearranging the last inequality, we get the first part of the third item. Also note that the first inequality also implies that $h(V) \le h'(V) + 2^z t\cdot \eta(V) + 2^z \sum_{v \in V} g(v)^z$, via Proposition \ref{prop:distance-z}. Then, by recalling that $g(v) \le \diam(V)$ for all $v \in V$, the second part of the third item follows.
		
		The proof of the fourth item is analogous to that of the third item. In addition, we need to combine the inequalities from the first item of the observation. We omit the details. The fifth item follows from the fact that $U \subseteq V$, and via triangle inequality.
	\end{proof}
	
	Let $\eta = \eta(V)^z$, $M = 2^{2z+2} (\eta(V)^z + \diam(V)^z)$, and $\delta = \xi M /2$. Then, the second item of Observation \ref{obs:kz-h} implies that $\eta \le h(v) \le \eta + M$ for all $v \in V$. Then, Proposition \ref{prop:concentration} implies that, 
	\begin{align*}
		&\Pr \biggl[ \left| \frac{\sum_{v \in V} \cost(v, C)}{|V|} - \frac{\sum_{u \in U} \cost(u, C)}{|U|} \right| \ge \frac{\xi}{2} 2^z (\eta(V)^z) + \diam(V)^z) \biggl] 
		\\&= \Pr \left[ \left| \frac{h(V)}{|V|} - \frac{h(U)}{|U|} \right| \ge \delta \right] 
		\\&\le \lambda'.
	\end{align*}
	Thus, with probability at least $1-\lambda'$, we have that 
	\begin{equation}
		\left| \frac{h(V)}{|V|} - \frac{h(U)}{|U|} \right| \le \frac{\xi}{2} \cdot M \label{eqn:kz-ineq2}
	\end{equation}
	In the rest of the proof, we condition on this event, and assume that (\ref{eqn:kz-ineq2}) holds, and show that the inequality in the lemma holds with probability 1. First, consider,
	\begin{align}
		\frac{h'(U)}{|U|} - \frac{h'(V)}{|V|} &\le \frac{h(U)}{|U|} - \frac{h(V)}{|V|} + \frac{2^z \cdot t \cdot (\eta(V)^z + \diam(V)^z)}{|V|} \nonumber \tag{From Obs. \ref{obs:kz-h}}
		\\&\le \frac{\xi}{2} M + \frac{ t \cdot M }{|V|}  \tag{From (\ref{eqn:kz-ineq2})} \nonumber
		\\&\le \xi M \label{eqn:kz-ineq3}
	\end{align}
	where the last inequality follows from the assumption that $|V| \ge s' \ge \frac{4q}{\xi} \ge \frac{4t}{\xi}$.
	Now, consider
	\begin{align}
		&\frac{h'(V)}{|V|} - \frac{h'(U)}{|U|} \nonumber
		\\&\le \frac{h(V)}{|V|} - \frac{h(U)}{|U|} + \frac{2^z \cdot t\frac{|U|}{|V|} \cdot (\eta(U)^z + \diam(V)^z)}{|U|} \tag{From Obs. \ref{obs:kz-h}, Part 4}
		\\&\le \frac{\xi}{2} M + \frac{2^z \cdot t \cdot \eta(U)^z}{|V|} + \frac{2^z \cdot t \cdot  \diam(V)^z}{|V|} \tag{From (\ref{eqn:kmed-ineq2}}
		\\&\le \frac{\xi}{2} M + \frac{2^{2z} \cdot t \cdot (\eta(V)^z + \diam(V)^z) + t \cdot 2^z \cdot \diam(V)^z}{|V|} \tag{From Obs. \ref{obs:kmed-h}, Part 5}
		\\&\le \frac{\xi}{2} M + \frac{2^{2z+1} \cdot t \cdot (\eta(V)^z + \diam(V)^z)}{|V|} \tag{Since $|V| \ge s' \ge \frac{4q}{\xi} \ge \frac{4t}{\xi}$}
		\\&= \frac{\xi}{2} M + \frac{\xi}{2} M = \xi M \label{eqn:kz-ineq4}
	\end{align}
	
	Note that (\ref{eqn:kz-ineq3}) and (\ref{eqn:kz-ineq4}) hold with probability $1$, conditioned on the inequality (\ref{eqn:kz-ineq2}) holding, which happens with probability at least $1-\lambda'$. Therefore, the following inequality holds with probability at least $1-\lambda'$:
	\begin{align}
		\left| h'(V) - h'(U) \cdot \frac{|V|}{|U|} \right| &\le 2^{2z+2}\xi \cdot |V| \cdot (\diam(V)^z + \d(V, C)^z)
	\end{align}
	where we recall that $\eta(V) = \d(V, C)$.	The preceding inequality is equivalent to the inequality in the lemma, by recalling that $h'(V) = \deltam{t}(V, C)$, and $h'(U) \cdot \frac{|V|}{|U|} = \frac{|V|}{|U|} \cdot \deltam{t'}(U, C) = \deltamw{t'}(U, C)$, since the weight of every sampled point in $U$ is equal to $|V|/|U|$. This concludes the proof of the lemma.
\end{proof}

Next, we show the following claim.
\begin{claim} \label{cl:kz-sum}
	\begin{itemize}
		\item $\sum_{i, j} |X_{i, j}| (2^j R)^z \le (1+2^z) \cdot \deltam{0}(X, A) \le (1+2^z)\tau \cdot \opt(\cI)$.
		\item $\sum_{i, j} |X_{i, j}| \diam(X_{i, j})^z \le 2^{z}(1+2^z) \cdot \deltam{0}(X, A) \le 2^{2z+1} \cdot \tau \cdot \opt(\cI)$.
	\end{itemize}
\end{claim}
\begin{proof}
	%Fix an optimal solution $(X, Y^*, C^*)$, and let $O^* = X \setminus Y^*$ denote the set of $m$ outliers in this solution. For each $X_{i, j}$, let $Y^*_{i, j} = X_{i, j} \cap Y^*$, and let $t_{i, j} = |X_{i, j} \cap O^*|$. 
	
	For any $p \in X_{i, j}$, it holds that $2^j R \le \max \LR{2\d(p, A), R} \le 2\d(p, A) + R$. 
	\begin{align*}
		\sum_{i, j} |X_{i, j}| \cdot (2^{j} R)^z  &\le \sum_{i, j} \sum_{p \in X_{i, j}} (2^j R)^z
		\\&\le \sum_{i, j} \sum_{p \in X_{i, j}} (2\d(p, A) + R)^z
		\\&= 2^z \sum_{p \in X} \d(p, A)^z + |X| \cdot R^z
		\\&= 2^z \cdot \deltam{0}(X, A) + n \cdot R^z
		\\&\le (1+2^z) \cdot \deltam{0}(X, A) \tag{By definition of $R$}
		\\&\le (1+2^z) \cdot \tau \cdot \opt(\cI') 
		\\&\le (1+2^z) \cdot \tau \cdot \opt(\cI). \tag{From Obs. \ref{obs:kz-lb}}
	\end{align*}
	We also obtain the second item by observing that $\diam(X_{i, j}) \le 2 \cdot 2^j \cdot R$, and using an analogous argument.
\end{proof}

Next, we show that the following lemma, which informally states that the union of the sets of sampled points approximately preserve the cost of clustering w.r.t.\ \emph{any} set of at most $k$ centers, \emph{even after} excluding at most $m$ outliers overall. 

\begin{lemma} \label{lem:kz-prob-bound}
	The following statement holds with probability at least $1-\lambda/2$:
	\\For all sets $C \subseteq F$ of size at most $k$, and for all sets of non-negative integers $\{m_{i, j}\}_{i, j}$ such that $\sum_{i, j} m_{i, j} \le m$, 
	\begin{align}
		&\left| \sum_{i, j} \deltam{m_{i, j}}(X_{i, j}, C) - \sum_{i, j} \deltamw{m'_{i, j}}(S_{i, j}, C) \right| \le \epsilon \cdot \sum_{i, j} \deltam{m_{i, j}}(X_{i, j}, C) \label{ineq:kz-coresetbound}
	\end{align}
	where $\displaystyle t_{i, j} = \left\lfloor m_{i, j} / w_{i, j}\right\rfloor$.
\end{lemma}
\begin{proof}
	Fix an arbitrary set $C$ of at most $k$ centers and the integers $\{m_{i,j}\}_{i, j}$ such that $\sum_{i, j} m_{i,j} \le m$ as in the statement of the lemma. For each $i = 1, \ldots, |A|$, and $0 \le j \le \phi$, we invoke Lemma \ref{lem:kz-samplinglemma} by setting $V \gets \xij$, and $U \gets \sij$, $\xi \gets \frac{\epsilon}{2^{9z} \tau}$, $\lambda' \gets n^{-k} \lambda / (4(k+m)(1+\phi))$, and $q \gets m$. This implies that, the following inequality holds with probability at least $1-\lambda'$ for each set $\xij$, and for the corresponding $m_{i, j} \le m$:
	\begin{align}
		&\left| \deltam{m_{i, j}}(X_{i, j},C) - \deltamw{t_{i,j}}(S_{i,j}, C) \right| \le \frac{\epsilon}{2^{9z}\tau} 2^{2z+2} |X_{i,j}| (\diam(X_{i,j}) + \d(X_{i, j}, C)) \label{eqn:kz-sumineq}
	\end{align}
	\iffalse
	where $q'_{i, j} = \lfloor q_{i, j}|U|/|V|\rfloor$. Note that the sample size required in order for this inequality to hold is 
	\[s' = \left\lceil \frac{4m}{\xi^2} \ln \lr{\frac{2}{\lambda}} \right\rceil = \left\lceil 4m\lr{\frac{2^{9z}\tau}{\epsilon}}^2 \cdot \ln\lr{\frac{8n^k(k+m)(1+\phi)}{\lambda}}\right\rceil \le s.\]
	Now, let $X' \subseteq X$ denote a subset of $m$ points that is farthest from the centers from $C$ (breaking ties arbitrarily) i.e., if the set of centers is fixed to be $C$, then $X'$ consists of $m$ outlier points. For each $i, j$, let $t_{i, j} = |X_{i, j} \cap X'|$, and $t'_{i, j} = \lfloor t_{i, j}|U|/|V|\rfloor$. Note that for any $i, j$, $0 \le t'_{i, j} \le t_{i, j} \le m$, and $\sum_{i, j} t'_{i, j} \le \sum_{i, j} t_{i, j} = m$. Then, note that $\Delta(C) = \sum_{i, j} \deltamw{t'_{i, j}}(S_{i, j}, C)$ satisfies all the conditions as stated in the statement of the lemma. Therefore, 
	\begin{align}
		\left| \deltam{m}(X, C) - \Delta(C) \right| &\le \sum_{i, j} \left| \deltam{t_{i, j}}(X_{i, j}, C) - \deltamw{t'_{i, j}}(S_{i, j}, C) \right| \label{eqn:kz-sumineq2}
	\end{align}
	\fi
	Note that for any $i, j$, if $\xij < s$, i.e., $\xij$ is \emph{small}, then the sample $S_{i,j}$ is equal to $X_{i, j}$, and each point in $S_{i, j}$ has weight equal to $1$. This implies that $\deltam{t_{i, j}}(X_{i, j}, C) = \deltamw{t'_{i, j}}(S_{i, j}, C)$ for all such $X_{i, j}$, the contribution to the right hand side of inequality (\ref{eqn:kz-sumineq}) is zero. Thus, it suffices to restrict the sum on the right hand side of (\ref{eqn:kz-sumineq}) over \emph{large} sets $X_{i, j}$'s. We have the following claim about the large sets $X_{i, j}$, the proof of which is analogous to that of Claim \ref{cl:kmed-sum-large}, and is therefore omitted.
	\begin{claim} \label{cl:kz-sum-large}
		$\sum_{i, j: X_{i, j} \text{ is large}} \d(X_{i, j}, C)^z \le 2\deltam{m}(X, C)$.
	\end{claim}
	Thus, by revisiting (\ref{eqn:kz-sumineq}), we get:
	\begin{align}
		&\left| \deltam{m}(X, C) - \Delta(C) \right| \nonumber
		\\&\le \sum_{i, j: X_{i, j} \text{ is large}} \left| \deltam{t_{i, j}}(X_{i, j}, C) - \deltamw{t'_{i, j}}(S_{i, j}, C) \right|
		\\&\le \frac{\epsilon}{2^{9z}\tau} \sum_{i, j: X_{i, j} \text{ is large}} 2^{2z+2} \cdot |X_{i,j}| (\diam(X_{i,j})^z + \d(X_{i, j}, C)^z) \tag{By setting $q_{i,j} \gets t_{i,j}$ and $q'_{i,j} \gets t'_{i,j}$ in (\ref{eqn:kz-sumineq})}
		\\&\le \frac{\epsilon}{2^{9z}\tau} \cdot (2^{2z+2} \cdot 2^{2z+1} \tau \opt(\cI) + 2^{2z+3}\deltam{m}(X, C) \tag{From Claim \ref{cl:kz-sum} and Claim \ref{cl:kz-sum-large}}
		\\&\le \frac{\epsilon}{2^{9z}\tau} (2^{9z} \cdot \tau \cdot \deltam{m}(X, C)) = \epsilon \cdot \deltam{m}(X, C)) \nonumber
	\end{align}
	Where, the last inequality follows from the fact that since $C$ is an arbitrary set of at most $k$ centers, $\opt(\cI) \le \deltam{m}(X, C)$.	Note that the preceding inequality holds for a fixed set $C$ of centers with probability at least $1 - |A|\cdot (1+\phi)\lambda' = 1-n^{-k}\lambda/2$, which follows from taking the union bound over all sets $X_{i, j}$, $1 \le i \le |A| \le k+m$, and $0 \le j \le \phi$. 
	
	Since there are at most $n^k$ subsets $C$ of $F$ size at most $k$, the statement of the lemma follows from taking a union bound.
	\iffalse 
	\begin{observation}
		For any set $P$ with $|P| > s \ge 2m$, the following holds:
		\begin{enumerate}
			\item For any $0 \le m' \le m$, \ \ $|P|/2 \le |P| - m \le |P|$
			\item For any $0 \le m_1, m_2 \le m$,\ \  $|P|-m_1 \le |P| \le 2(|P|-m_2)$.
		\end{enumerate}
	\end{observation}
	\fi
\end{proof}

Once we obtain a coreset $S$ satisfying Lemma \ref{lem:kz-prob-bound}, we can perform a similar enumeration of sets of size at most $m$, and obtain $\lr{\frac{k+m}{\epsilon}}^{O(m)} \cdot n^{O(1)}$ instances of $(k, z)$-\textsc{Clustering}. We call a $\beta$-approximation on each of these instances, and each call takes time $T(n, k)$. The subsequent analysis is identical to that for \kmedwo which can be used to show an analogous version of Lemma \ref{lem:final-kmed-bound}. We omit the details, and conclude this section with the following theorem, which generalizes Theorem \ref{thm-generalreduction}.

\begin{theorem}
	Let $z \ge 1$ be a fixed constant. Suppose there exists a $\beta$-approximation algorithm for $(k, z)$-\textsc{Clustering} with running time $T(n,k)$ for some constant $\beta \geq 1$, and there exists a $\tau$-approximation algorithm for $(k, z)$-\textsc{Clustering} that runs in polynomial time, where $\tau = O(1)$. Then there exists a $(\beta+\epsilon)$-approximation algorithm for $(k, z)$-\textsc{Clustering with Outliers}, with running time $\lr{\frac{k+m}{\epsilon}}^{O(m)} n^{O(1)} \cdot T(n,k)$, where $n$ is the instance size and $m$ is the number of outliers.
\end{theorem}

\end{document}